\RequirePackage{fix-cm}
\documentclass[a4paper,11pt]{article}

\usepackage{authblk}
\usepackage{amsfonts,amsmath,amssymb,amsthm}
\usepackage{paralist}
\usepackage{url}
\usepackage{hyperref}
\usepackage{graphicx}
\usepackage{microtype}
\usepackage{subcaption}
\usepackage{url}
\usepackage{hyperref}

\newcommand{\N}{\mathbb{N}}

\newcommand{\F}{\mathbb{F}}

\newtheorem{definition}{Definition}
\newtheorem{lemma}{Lemma}
\newtheorem{theorem}{Theorem}
\newtheorem{problem}{Problem}

\providecommand{\keywords}[1]{\textbf{\textit{Keywords }} #1}

\begin{document}

\title{Hip to Be (Latin) Square: Maximal Period Sequences from Orthogonal Cellular Automata}

\author[1]{Luca Mariot}
	
\affil[1]{{\small Cyber Security Research Group, Delft University of Technology, Mekelweg 2, Delft, The Netherlands} 
	
	{\small \texttt{l.mariot@tudelft.nl}}}

\maketitle

\begin{abstract}
Orthogonal Cellular Automata (OCA) have been recently investigated in the literature as a new approach to construct orthogonal Latin squares for cryptographic applications such as secret sharing schemes. In this paper, we consider OCA for a different cryptographic task, namely the generation of pseudorandom sequences. The idea is to iterate a dynamical system where the output of an OCA pair is fed back as a new set of coordinates on the superposed squares. The main advantage is that OCA ensure a certain amount of diffusion in the generated sequences, a property which is usually missing from traditional CA-based pseudorandom number generators. We study the problem of finding OCA pairs with maximal period by first performing an exhaustive search up to local rules of diameter $d=5$, and then focusing on the subclass of linear bipermutive rules. In this case, we characterize the periods of the sequences in terms of the order of the subgroup generated by an invertible Sylvester matrix. We finally devise an algorithm based on Lagrange's theorem to efficiently enumerate all linear OCA pairs of maximal period up to diameter $d=11$.
\end{abstract}


\section{Introduction}
\label{sec:intro}
Cellular Automata (CA) have been extensively used in the past to define cryptographic primitives, especially \emph{Pseudorandom Number Generators} (PRNGs). Indeed, CA are an interesting computational model for generating pseudorandom sequences, since they can exhibit very chaotic dynamical behaviors. Moreover, the massive parallelism inherent to CA lends itself to efficient hardware implementations. Nevertheless, CA-based PRNGs such as those based on rule 30 pioneered by Wolfram~\cite{wolfram85} have later been found insecure, since an adversary can efficiently recover the initial state of the CA by observing only the trace of the cell sampled as a pseudorandom sequence~\cite{meier91,koc97}. Later research~\cite{martin08,formenti14,leporati13,leporati14} focused on improving the security of Wolfram-like PRNGs by investigating local rules of higher radii with better cryptographic primitives, especially related to the \emph{confusion} principle set forth by Shannon~\cite{shannon49}. Still, another problem that has received little attention in this research thread is that CA in general also has poor \emph{diffusion}, meaning that the differences between distinct initial states propagate too slowly in the dynamic evolution of the CA~\cite{daemen94}. This flaw is mostly due to the local nature of the CA update rule, and it can represent a problem with respect to \emph{differential cryptanalysis} attacks.

In this work, we investigate a new method for generating pseudorandom sequences by cellular automata, based on the iteration of \emph{Orthogonal} CA (OCA). Orthogonal CA are pairs of CA whose superposed global rules generate \emph{orthogonal Latin squares}, and up to now they have been mostly analyzed in connection with \emph{secret sharing schemes}~\cite{mariot18,mariot20}. The main idea underlying our method is to define a dynamical system whose state is the input configuration of an OCA pair. Then, the system is iterated by concatenating the output of the two OCA as a new input configuration. Intuitively, this process starts from a random cell on the orthogonal Latin squares, and uses the superposed entries as the coordinates of the new cell where to "jump" in the next iteration.

The motivation of our work is twofold. First, dynamical systems arising from OCA are reversible, which is useful in cryptographic applications such as block ciphers. Second, the orthogonality of the corresponding Latin squares allows to implement a $(2,2)$-multipermutation~\cite{vaudenay94}, which guarantees a certain amount of diffusion between blocks of $d-1$ cells, where $d$ is the diameter of the local rules.

A desirable property for a PRNG is to feature a large period, starting from any seed. For this reason, after giving the necessary background definitions in Section~\ref{sec:prel} and defining the dynamical system in Section~\ref{subsec:desc-gen}, we perform an exhaustive search of OCA pairs up to diameter $d=5$ to compute their cycle decompositions, remarking that the maximal period $2^{2n}-1$ is attained only by pairs formed by linear rules. Subsequently, in Section~\ref{sec:char-loca} we characterize the periods of linear OCA pairs as the order of the subgroup generated by the associated Sylvester matrix. Next, leveraging on Lagrange's theorem, we prove that the maximum order is indeed upper bounded by $2^{2n}-1$. We finally design an algorithm to efficiently enumerate all linear OCA pairs of maximal periods based on our theoretical results, and apply it up to diameter $d=11$. Such findings cue us to several open problems and further directions of research on this subject, which we discuss in the conclusions of the paper in Section~\ref{sec:outro}.


\section{Background}
\label{sec:prel}
We start by giving some basic definitions related to cellular automata and orthogonal Latin squares used in the remainder of the paper. In what follows, by $[N] = \{1,\cdots,N\}$ we denote the set of the first $N$ positive integers, while $\F_q$ stands for the \emph{finite field} with $q$ elements, where $q$ is a power of a prime number. Further, the $n$-dimensional vector space over $\F_q$ is denoted by $\F_q^n$.

Cellular automata are a parallel computational model whose global state is described by an array of \emph{cells}, usually arranged on a line or a grid. Each cell synchronously updates its state by evaluating a local update rule on itself and a certain amount of neighboring cells. In this work, we are mainly interested in the model of one-dimensional \emph{No-Boundary CA} (NBCA), studied in~\cite{mariot19,mariot20} respectively in the context of S-boxes and orthogonal Latin squares:
\begin{definition}
\label{def:ca}
A \emph{No-Boundary CA} is a vectorial function $F:\F_q^n \to \F_q^{n-d+1}$ defined by a \emph{local rule} $f: \F_q^d \to \F_q$ with \emph{diameter} $d\le n$, such that
\begin{equation}
\label{eq:glob-ca}
F(x_1,\cdots,x_n) = (f(x_1,\cdots,x_d),\cdots,f(x_{n-d+1},\cdots,x_n))
\end{equation}
for all $x = (x_1,\cdots,x_n) \in \F_q^n$.
\end{definition}
In other words, each output coordinate is determined by evaluating the local rule $f$ on the \emph{neighborhood} formed by the $i$-th input cell and the $d-1$ cells to its right. The output vector is shorter than the input, since the local rule is applied only until the coordinate $n-d+1$ to avoid boundary conditions (which is the reason why this model is called \emph{No-Boundary}). One of the most studied settings are CA over the binary alphabet (i.e. $q=2$), where the local rule is a \emph{Boolean function} $f: \F_2^d \to \F_2$ of $d$ variables. In this case, it is common to use \emph{Wolfram's convention} to encode a CA local rule, which is basically the decimal encoding of its $2^d$-bit truth table~\cite{wolfram83}.

A \emph{Latin square} of order $N \in \N$ is a $N \times N$ matrix $L$ with entries from $[N]$ such that every row and every column are permutations of $[N]$. Two Latin squares $L_1$ and $L_2$ of order $N$ are called \emph{orthogonal} if
\begin{equation}
  (L_1(i_1,j_1),L_2(i_1,j_1)) \neq (L_1(i_2,j_2),L_2(i_2,j_2))
\end{equation}
for all distinct pairs of coordinates $(i_1,j_1), (i_2,j_2) \in [N]\times [N]$. Stated otherwise, two Latin squares are orthogonal if their \emph{superposition} yields all
the ordered pairs of the Cartesian product $[N] \times [N]$. Orthogonal Latin squares have several applications in cryptography, most notably related to \emph{secret sharing schemes}~\cite{stinson04}.

In~\cite{mariot20}, the authors showed how to generate orthogonal Latin squares with cellular automata, which have later been named \emph{orthogonal CA} (OCA) in~\cite{mariot18}. The basic idea is to consider CA with bipermutive local rules. More precisely, a function $f: \F_q^d \to \F_q$ is called \emph{bipermutive} if, by fixing either the leftmost or the righmost $d-1$ input variables to any value, the resulting restriction on the remaining coordinate is a permutation of $\F_q$. Eloranta~\cite{eloranta93} proved that a NBCA $F: \F_q^{2(d-1)} \to \F_q^{(d-1)}$ with bipermutive local rule $f: \F_q^d \to \F_q$ gives rise to a Latin square of order $N=q^{d-1}$, a result which was later independently rediscovered in~\cite{mariot16}. The idea is to use the left and right halves of the input configuration to index respectively the row and the column of the square, and the output configuration as the entry at those coordinates. The characterization of OCA given in~\cite{mariot20} considers bipermutive local rules that are also \emph{linear}, i.e. $f: \F_q^d \to \F_q$ is defined for all $x \in \F_q^d$ as a linear combination $f(x_1,\cdots,x_d) = a_1x_1 + \cdots + a_dx_d$, with $a_i \in \F_q$ for $i \in [d]$ and the constraint that $a_1, a_d \neq 0$ to ensure bipermutivity. A polynomial of degree $n=d-1$ with coefficients in $\F_q$ can be naturally associated to a linear rule, by using the mapping $f \mapsto P_f(X) = a_1 + a_2X + \cdots + a_dX^n$. Then, the characterization of linear OCA proved in~\cite{mariot20} can be stated as follows:
\begin{theorem}
\label{thm:lin-oca}
Let $F,G: \F_q^{2n)} \to \F_q^{n}$ be two NBCA defined by linear bipermutive local rules $f,g: \F_q^d \to \F_q$ of diameter $d$, with $n=d-1$. Then, the two Latin squares of order $N = q^{n}$ generated by F and G are orthogonal if and only if the polynomials $P_f(X), P_g(X) \in \F_q[X]$ of degree $n$ respectively associated to $f$ and $g$ are relatively prime.
\end{theorem}

\section{Generating Sequences with Orthogonal CA}
\label{sec:gen-oca}
In this section we first define the dynamical system used to generate sequences with OCA, and formally state the problem of identifying the local rules that induce maximum length cycles. We then exhaustively enumerate such rules up to diameter $d=5$.

\subsection{Description of the Generator and Problem Statement}
\label{subsec:desc-gen}
As showed in Section~\ref{sec:prel}, any bipermutive CA with local rule $f$ of diameter $d$ defines a Latin square of order $q^n$, where $n=d-1$ and $q$ is the size of the alphabet. However, one cannot use such a CA to directly generate a pseudorandom sequence, as done in Wolfram-like PRNGs. Indeed, since the cellular automaton is in the no-boundary setting and the initial configuration is composed of $2n$ cells, only a configuration of length $n$ results from a single evaluation of the global rule, leaving not enough cells for a second iteration.

Instead of using a single local rule, the main idea behind our generator is to consider a \emph{pair} of local rules $f,g: \F_q^d \to \F_q$, both applied to the same initial configuration $s$ of length $2n = 2(d-1)$, as in the case of orthogonal Latin squares. In this way, one gets two output configurations $z = F(s)$, $w = G(s)$, each of length $n$, generated by the global rules $F,G$ respectively induced by $f$ and $g$. At this point, we construct a new configuration of length $2n$ by \emph{concatenating} $z$ and $w$. Therefore, the outputs of the NBCA $F,G$ are used respectively as a new row and a new column coordinate, which will in turn point to a new pair of entries given by $F$ and $G$. Seen on the superposed Latin squares generated by $F$ and $G$, this process can be visualized as starting from the pair of entries indexed by the initial configuration $s$, and then using such entries as the destination coordinates where to "jump" in the next step.

We can now give a formal definition of the dynamical system $\mathcal{S}$ intuitively described above. Given $d \in \N$ and $q$ a power of a prime, the phase space of $\mathcal{S}$ is the vector space $\F_q^{2n}$ where $n=d-1$, i.e. the set of all vectors of length $2n$ over the finite field $\F_q$. In particular $\F_q^{2n}$ is isomorphic to the Cartesian product $\F_q^n \times \F_q^n$, the set of all ordered pairs of $n$-dimensional vectors over $\F_q$. Slightly abusing notation, in what follows we will also consider a pair of vectors of $n$ components as an element of $\F_q^{2n}$. In fact, going from one representation to the other simply entails adding and dropping parentheses accordingly.

Let $s(0) = (x(0), y(0)) \in \F_q^{2n}$ be the initial state of the system $\mathcal{S}$. Further, let $f,g: \F_q^d \to \F_q$ be two bipermutive local rules of diameter $d$, with $F,G: \F_q^{2n} \to \F_q^{n}$ being the corresponding NBCA global rules. Finally, denote by $H: \F_q^{2n} \to \F_q^{2n}$ the function that \emph{concatenates} the results of the two global rules $F$ and $G$ to the same CA input $x \in \F_q^{2n}$. Then, the system $\mathcal{S}$ updates its current state $s(t)$ at time $t \in \N$ through the following equation:
\begin{equation}
\label{eq:sys-upd}
s(t+1) = H(x(t+1), y(t+1)) = (F(s(t)), G(s(t))) \enspace .
\end{equation}
In other words, the state of the system is always separated in two equal-size parts, where the left part comes from the application of the first global rule on the whole state in the previous step, whereas the right part is defined analogously as the result of the second global rule evaluated on the previous state. Figure~\ref{fig:dynsys} depicts the block diagram for the dynamical evolution of the system. In what follows, we will compactly denote such a system $\mathcal{S}$ by the pair $\langle \F_q^{2n}, H \rangle$.
\begin{figure}[t]
    \centering
    \includegraphics[height=9cm]{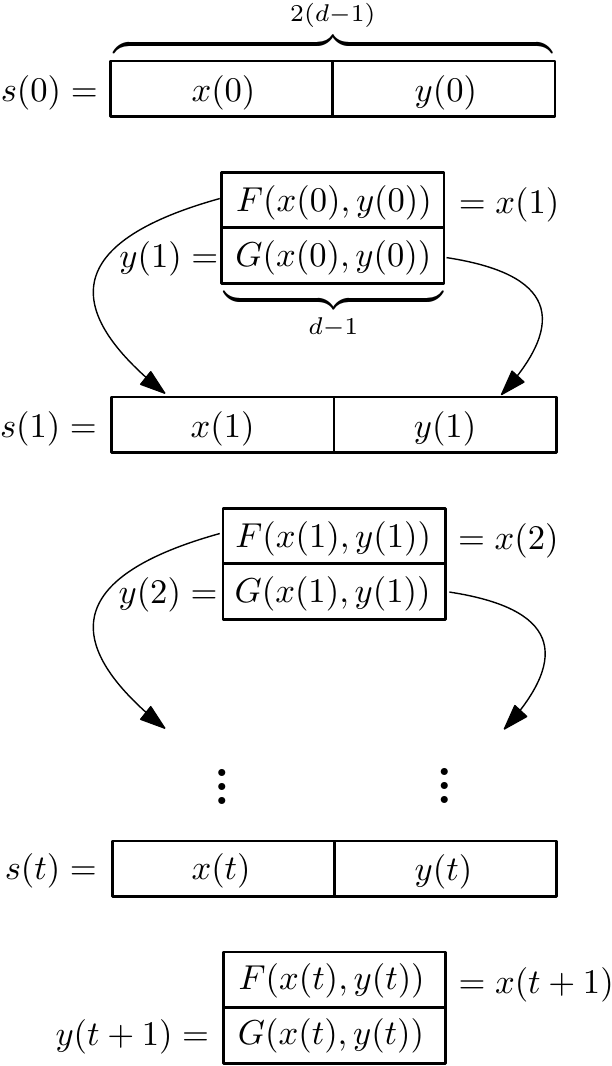}
    \caption{block diagram for the dynamical evolution of the system starting from the initial state $s(0) = (x(0), y(0)) \in \F_2^{2n}$.}
    \label{fig:dynsys}
\end{figure}
In principle, one could sample the orbit arising from the iteration of Equation~\eqref{eq:sys-upd} as a pseudorandom sequence, starting from a random initial configuration $s(0)$. However, pseudorandom sequences adopted in domains such as cryptography need to satisfy several stringent properties, meaning that one cannot simply select a pair of local rules at random. For this reason, beside choosing $f$ and $g$ to be bipermutive local rules, we also require that the Latin squares generated by the global rules $F$ and $G$ are \emph{orthogonal}. The motivation is twofold:
\begin{compactenum}
\item As recalled in Section~\ref{sec:prel}, a pair of orthogonal Latin squares of order $N$ defines a \emph{permutation} over the Cartesian product $[N]\times [N]$. It follows that the update function defined in Equation~\eqref{eq:sys-upd} is bijective. Thus, the resulting system is \emph{reversible}, or equivalently its trajectories are all disjoint cycles, without transient parts. In practice, reversibility implies that the system can also be run backward in time, by applying the inverse permutation. Such a property is important in certain cryptographic primitives (e.g., SPN block ciphers) where, beside generating pseudorandom sequences, there is also the need of inverting the global state of the cipher to ensure decryption. In the particular setting of OCA, one could invert the system by using the algorithm based on coupled de Bruijn graphs described in~\cite{mariot18}.
\item Orthogonal Latin squares coincide with a particular kind of \emph{Maximum Distance Separable (MDS) codes}, which are of great importance in the design of \emph{diffusion layers} for block ciphers. The reason is that layers based on MDS codes spread the statistical structure of the plaintext over the ciphertext in an optimal way, providing resistance against differential cryptanalysis. In particular, as shown by Vaudenay~\cite{vaudenay94}, the function $H$ defined in Equation~\eqref{eq:sys-upd} corresponds to a $(2,2)-$\emph{multipermutation}, i.e. any distinct pair of input/output tuples $(x,y, F(x,y), G(x,y))$ and $(x',y',F(x',y'),G(x',y'))$ \emph{cannot agree} on any $2$ coordinates. Stated differently, such tuples must be at Hamming distance at least $3$.
\end{compactenum}
In this work we investigate the dynamics of the system $\mathcal{S} = \langle \F_q^{2n}, H \rangle$ when the underlying local rules $f$ and $g$ generate a pair of OCA. More precisely, we are interested in studying the periods of the cycles in $\mathcal{S}$. Given a state $s \in \F_q^{2n}$, the (minimum) \emph{period} of $s$ under $\mathcal{S}$ is the smallest positive integer $p$ such that $H^{p}(s) = s$. In other words, $p$ is the smallest number of iterations of $H$ after which the state of the system returns to the initial condition $s$. In cryptography, pseudorandom sequences with very large periods are usually sought. This is due to the fact that  cryptographic primitives such as \emph{stream ciphers} encrypt the plaintext by computing the bitwise XOR between the plaintext and a \emph{keystream}, which is actually a pseudorandom sequence generated by stretching a short secret key~\cite{stinson18}. In particular, an attacker can mount certain attacks based on frequency analysis when the pseudorandom sequences used as keystreams have a period that is shorter than the plaintext. Ideally, the dynamics of a pseudorandom generator used in cryptography should be composed of a single large cycle that visits all states in the phase space.

We now state the problem addressed in the rest of the paper:
\begin{problem}
\label{pb:stat}
Let $d \in \N$ and $q$ be a power of a prime number, and let $n=d-1$. What is the maximal period attainable by the system $\mathcal{S} = \langle \F_q^{2n}, H \rangle$, with $H$ defined as in Equation~\eqref{eq:sys-upd}, when the bipermutive local rules $f,g: \F_q^d \to \F_q$ induce a pair of orthogonal CA?
\end{problem}

We will mainly consider the binary alphabet $\F_2$ (i.e., $q=2$), since this is the simplest case to analyze and also the most useful one for cryptographic applications. However, most of the theoretical results presented in the next sections can be straightforwardly lifted to OCA over any finite field $\F_q$.

\subsection{Empirical Search up to $d=5$}
\label{subsec:prel-res}
We started our investigation of Problem~\ref{pb:stat} by conducting an exhaustive search over all pairs of bipermutive rules up to diameter $d=5$, filtering only those that yield OCA and analyzing the cycle decompositions of the corresponding dynamical systems. For each diameter $d\le 6$, Table~\ref{tab:sizes} reports the size of the CA output configuration $n=d-1$, the order of the corresponding Latin squares $2^n$, the number of bipermutive local rules $\mathcal{B}_d = 2^{2^{d-2}}$, the number of ordered pairs that can be formed with them $\mathcal{B}_d^2 = 2^{2^{d-2}}\cdot 2^{2^{d-2}} = 2^{2^n}$, and finally the number of pairs which generate OCA.
\begin{table}[t]
\centering
\caption{Size of the search space $\mathcal{B}^2_d$ of bipermutive rules pairs with diameter $d$.}
\begin{tabular}{p{1cm}p{1cm}p{1cm}p{1cm}p{1cm}p{1cm}}
\hline\noalign{\smallskip}
$d$ & $n$ & $2^{n}$ & $\mathcal{B}_d$ & $\mathcal{B}^2_d$ & $\mathcal{OCA}_d$ \\
\noalign{\smallskip}\hline
2 & 1 & 2 & 2 & 4 & 0 \\
3 & 2 & 4 & 4 & 16 & 8 \\
4 & 3 & 8 & 16 & 256 & 72 \\
5 & 4 & 16 & 256 & 65536 & 1704 \\
6 & 5 & 32 & 65536 & $4.3 \cdot 10^9$ & 533480 \\
\hline
\end{tabular}
\label{tab:sizes}
\end{table}
The last column has been taken from~\cite{mariot17}, where an exhaustive search of all OCA up to diameter $d=6$ is performed by means of a combinatorial algorithm that enumerates only pair of bipermutive local rules which are \emph{pairwise balanced}. As a matter of fact, up to now there are no known methods to enumerate generic OCA pairs, unless one narrows the attention to the case of linear rules addressed in~\cite{mariot20}. Further, remark that the numbers in the last column of Table~\ref{tab:sizes} multiplies by $8$ the numbers of OCA pairs reported in~\cite{mariot17}, since we did not consider any symmetry relation preserving the orthogonality property as done in that work.

The size of the search space of interest for our empirical investigation is thus specified by the fourth column of Table~\ref{tab:sizes}, $\mathcal{B}_d^2$. In particular, our exhaustive search enumerated all ordered pairs of bipermutive local rules of diameter $d$, selected only those that generated OCA, and for each of them determined the cycle decomposition of the corresponding dynamical system $\mathcal{S}$ defined in Section~\ref{subsec:desc-gen}. In principle, it is also possible to extend such search to diameter $d=6$, since the size of the resulting space ($\mathcal{B}_2^d \approx 4.3 \cdot 10^9$) is still amenable to exhaustive enumeration in a reasonable time. However, in our experiments we limited our search up to $d=5$ since this was enough to inform our theoretical investigation described in the next section.

The case of diameter $d=2$ can be immediately discarded, since no OCA pairs exist with this parameter. In fact, one can easily see that there are only two Latin squares of order $2^{2-1} = 2$, and they are not orthogonal. For diameter $d=3$, a total of $8$ OCA pairs result from the search over all 16 pairs of bipermutive rules. \emph{All these OCA pairs resulted in the same cycle decomposition structure}, i.e. one fixed point and a single cycle of length $15$.
As an example, Figure~\ref{fig:oca-90-150} reports the cycle decomposition of the OCA pair formed by the rules with Wolfram codes 90 and 150 respectively, along with the associated paths on the superposed squares.
\begin{figure}[t]
\centering
\subfloat[Cycle decomposition]{\includegraphics[width=0.5\textwidth]{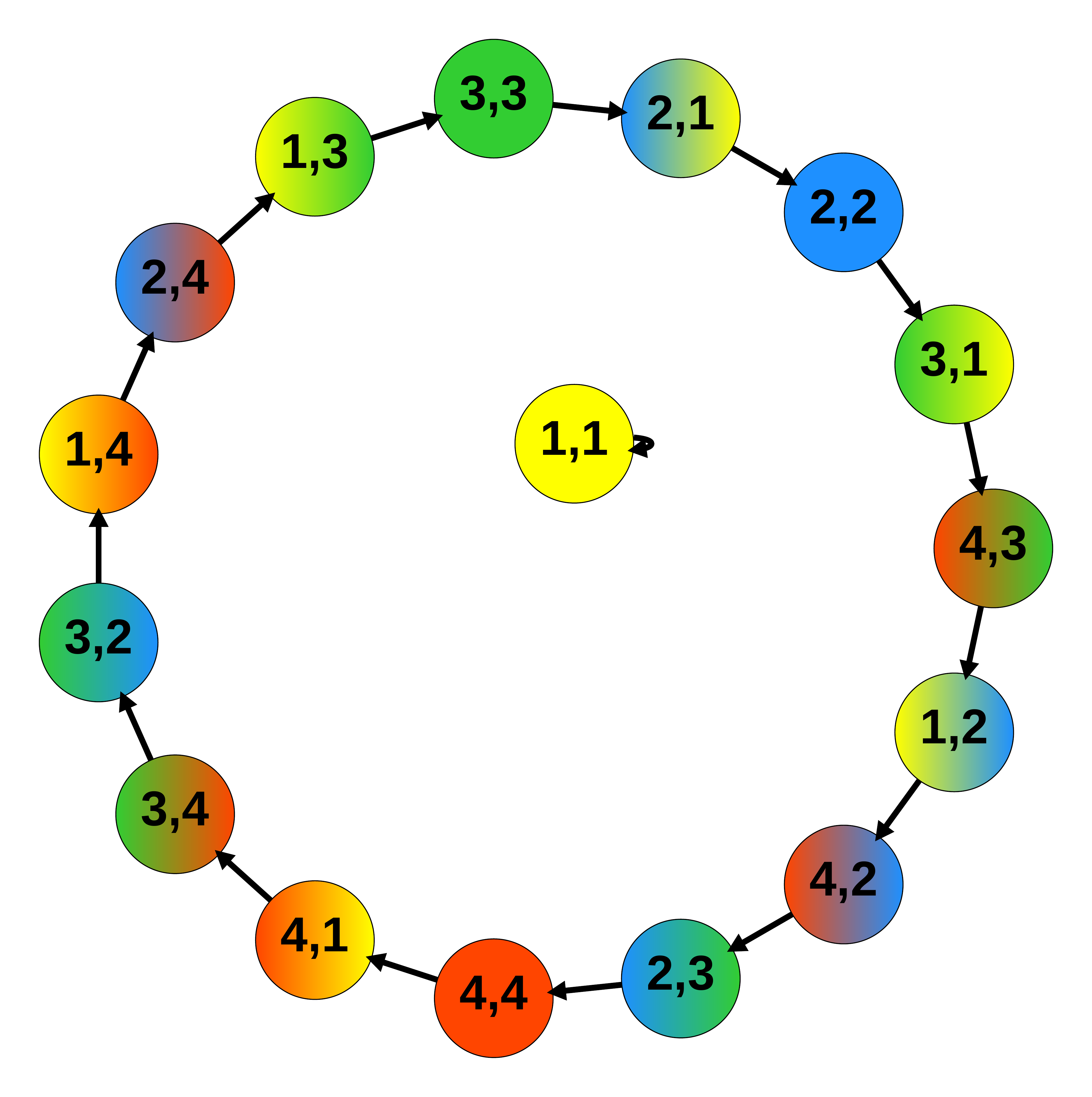}}%
\phantom{MM}%
\subfloat[Paths on superposed squares]{
\raisebox{0.8cm}{
\includegraphics[width=0.35\textwidth]{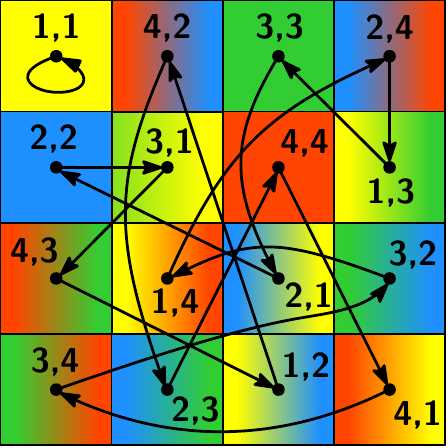}}
}
\caption{Example of cycle decomposition and paths on the squares generated by the OCA with local rules 90 and 150.}
\label{fig:oca-90-150}
\end{figure}
Consequently, all 8 OCA pairs of diameter $d=3$ feature a maximum cycle length which is equal to the area of the square ($2^{2\cdot 2} = 16$) minus 1, or equivalently, there is a single walk on the superposed squares that visits all cells except one (i.e., the fixed point).
 \begin{figure}[b]
 \centering
 \includegraphics[width=0.95\textwidth]{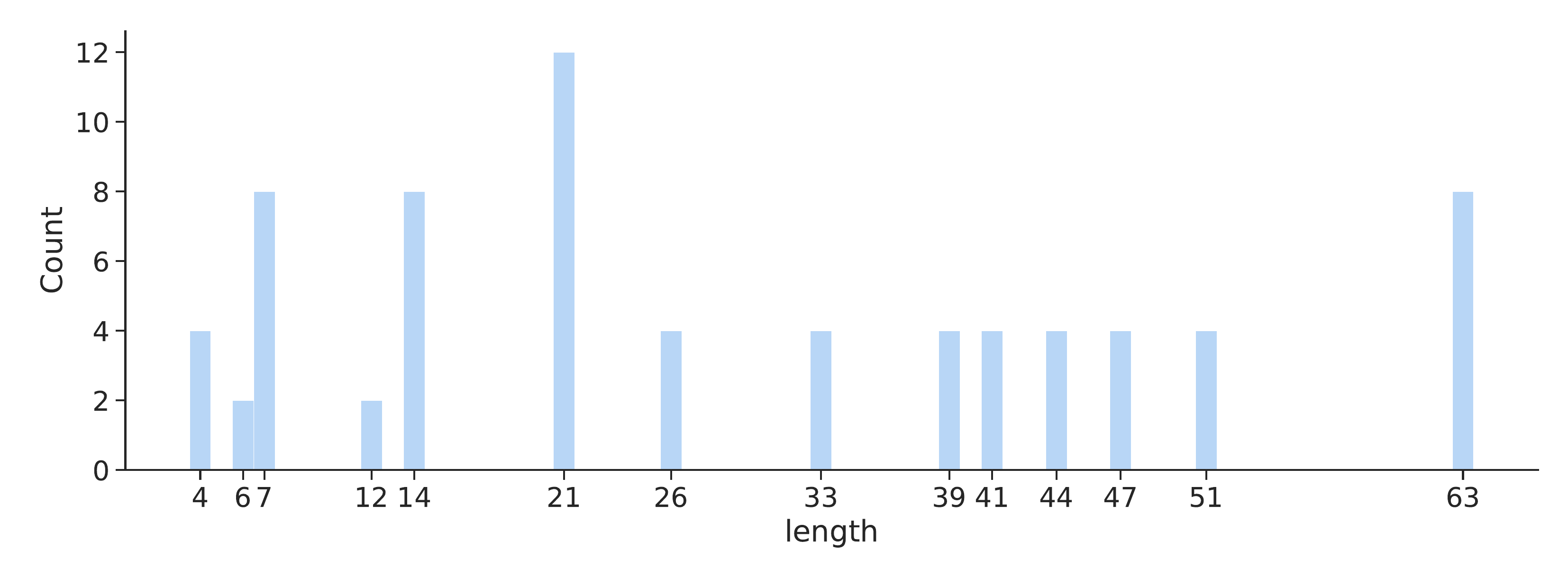}
 \caption{Distribution of maximum cycle lengths for OCA of diameter $d=4$.}
 \label{fig:distr}
\end{figure}
Similar conclusions can be drawn also from the results for $d=4$ and $d=5$, with Figure~\ref{fig:distr} reporting the distribution of the maximum cycle lengths for OCA with $d=4$ as an example. In particular, \emph{no OCA pair is able to attain the $2^{2n}$ upper bound on the maximum cycle length}. In other words, there is no OCA featuring a single ``pure cycle'' that visits all cells in the superposed squares. Rather, the best decomposition possible is a single fixed point and a cycle of length $2^{2n}-1$. This almost optimal situation is achieved by 8 OCA pairs for $d=4$, whose largest cycle has length $63$, and 36 pairs for $d=5$, with a maximum length cycle of $255$. A closer inspection of the types of local rules forming such OCA leads us to the second interesting finding: \emph{all OCA pairs reaching a maximum cycle length of $2^{2n}-1$ are defined by linear local rules}. For this reason, in the remainder of this paper we consider only OCA pairs defined by linear rules.

\section{The Case of Linear OCA}
\label{sec:char-loca}
We now delve into the case of linear OCA pairs, providing a theoretical characterization of their periods. As it often happens when studying the behavior of dynamical systems governed by a linear transformation, such a characterization is made possible by the use of linear algebra methods.

Let $f,g: \F_2^d \to \F_2$ be two linear bipermutive local rules of diameter $d$. Following the notation recalled in Section~\ref{sec:prel}, we assume that the linear combinations defining $f$ and $g$ are respectively given by the two vectors $a = (a_1, \cdots, a_d) \in \F_2^d$ $b = (b_1, \cdots, b_d) \in \F_2^d$, where $a_1=b_1=a_d=b_d=1$ to ensure bipermutivity. Let $P_f(X), P_g(X) \in \F_2[X]$ be the monic polynomials of degree $n=d-1$ and nonzero constant term associated to $f$ and $g$. Then, by Theorem~\ref{thm:lin-oca} $f$ and $g$ induce a pair of OCA if and only if their polynomials $P_f(X)$ and $P_g(X)$ are relatively prime. As proved in~\cite{mariot20}, this characterization stands on the fact that the transformation which associates the CA initial configuration $(x,y) \in \F_2^{2n}$ to the pair of outputs $F(x,y), G(x,y)$ is defined by the following $2n \times 2n$ matrix:

\begin{equation}
    \label{eq:sylv-matr}
    M_{f,g} =
    \begin{pmatrix}
      a_1    & \cdots & a_{d} & 0 & \cdots & \cdots & \cdots & \cdots & 0 \\
      0      & a_1    & \cdots  & a_{d} & 0 & \cdots & \cdots & \cdots & 0 \\
      \vdots & \vdots & \vdots & \ddots  & \vdots & \vdots & \vdots & \ddots & \vdots \\
      0 & \cdots & \cdots & \cdots & \cdots & 0 & a_1 & \cdots & a_{d} \\
      b_1    & \cdots & b_{d} & 0 & \cdots & \cdots & \cdots & \cdots & 0 \\
      0      & b_1    & \cdots  & b_{d} & 0 & \cdots & \cdots & \cdots & 0 \\
      \vdots & \vdots & \vdots & \ddots  & \vdots & \vdots & \vdots & \ddots & \vdots \\
      0 & \cdots & \cdots & \cdots & \cdots & 0 & b_1 & \cdots & b_{d} \\
\end{pmatrix} \enspace .
\end{equation}
In particular, the two rules generate a pair of OCA if and only if the transformation $M_{f,g}\cdot (x,y)^\top$ is bijective, or equivalently if and only if $M_{f,g}$ is invertible. The matrix defined in Equation~\ref{eq:sylv-matr} has been extensively investigated in the literature: indeed, $M_{f,g}$ is the \emph{Sylvester matrix} associated to the two polynomials $P_f(X)$ and $P_g(X)$. It is a well known fact that the determinant of a Sylvester matrix---also called the \emph{resultant}---is not null if and only if $P_f(X)$ and $P_g(X)$ do not have any factor in common~\cite{gelfand08}. Therefore, the research in~\cite{mariot20} focused on counting the number of invertible Sylvester matrices defined by linear bipermutive rules, or equivalently on counting the number of linear OCA pairs.

The next lemma shows that computing the $t$-th iteration of the dynamical system $\mathcal{S}$ defined in Section~\ref{subsec:desc-gen} corresponds to multiplying the $t$-th power of the Sylvester matrix $M_{f,g}$ by the initial state vector, when the local rules are linear.
\begin{lemma}
\label{lm:t-iter}
Given $d \in \N$ and $n=d-1$, let $\mathcal{S} = \langle \F_2^{2n}, H \rangle$ be the dynamical system defined by the update function in Equation~\eqref{eq:sys-upd}, where the CA $F,G: \F_2^{2n} \to \F_2^n$ are defined by two bipermutive local rules $f,g: \F_2^d \to \F_2$ of diameter $d$ whose associated polynomials $P_f(X), P_g(X) \in \F_2[X]$ are coprime. Then, for any initial state $s(0) = (x(0),y(0)) \in \F_2^{2n}$, the state of $\mathcal{S}$ at time $t \in \N$ is given by:
\begin{equation}
\label{eq:t-iter}
s(t) = (x(t),y(t)) = M_{f,g}^t \cdot s(0) = M_{f,g}^t \cdot (x(0), y(0))^\top \enspace .
\end{equation}
\end{lemma}
\begin{proof}
We proceed by induction on $t \in \N$. The base case $t=1$ corresponds to the observation above on Theorem~\ref{thm:lin-oca}: a single application of the transformation $H: \F_2^{2n} \to~\F_2^{2n}$ defined in Equation~\eqref{eq:sys-upd} corresponds to the matrix-vector multiplication $M_{f,g} \cdot (x(0),y(0))^\top$. Let us assume now that the claim is valid for any $t \in \N$, and consider the case $t+1$: this is equivalent to iterating $H$ for $t+1$ steps starting from $s(0)$, which can be written equivalently as the composition of $H$ with its $t$-th iterate $H^t$:
\begin{equation}
\label{eq:t-iter2}
s(t+1) = H^{t+1}(s(0)) = H \circ H^t(s(0)) \enspace .
\end{equation}
By induction hypothesis, we know that $H^{t}(s(0)) = M_{f,g}^t \cdot s(0)^\top$, and that a single application of $H$ amounts to multiplying $M_{f,g}$ with the current state vector. Hence, we can rewrite Equation~\eqref{eq:t-iter2} as follows:
\begin{equation}
\label{eq:t-iter3}
H^{t+1}(s(0)) = H \circ H^t(s(0)) = M_{f,g} \cdot (M_{f,g}^t \cdot s(0)^\top)^\top \enspace ,
\end{equation}
from which we conclude that $s(t+1) = M_{f,g}^{t+1} \cdot s(0)^\top = M_{f,g}^{t+1} \cdot (x(0), y(0))^\top$. \qed
\end{proof}
Concerning Problem~\ref{pb:stat} when the underlying OCA are defined by a pair of linear local rules, Lemma~\ref{lm:t-iter} implies that the periods of the cycles in system $\mathcal{S}$ are determined by the \emph{order} of the associated Sylvester matrix $M_{f,g}$, considered as an element of the \emph{general linear group} $GL(2n, \F_2)$. The general linear group $GL(2n, \F_2)$ is defined as the set of all invertible matrices of size $2n \times 2n$ with entries in $\F_2$, equipped with matrix multiplication as a group operation. Indeed, the orthogonality requirement constrains $M_{f,g}$ to be invertible, and Lemma~\ref{lm:t-iter} establishes that the $t$-th iterate of the transformation $H$ corresponds to the $t$-th power of such matrix. Thus, determining the period of a state $s \in \F_2^{2n}$ is equal to finding the minimum $t \in \N$ such that $M_{f,g}^t = I_{2n}$, i.e. the $t$-th power of $M_{f,g}$ is the identity matrix of order $2n$. This is in turn equivalent to determining the order of the cyclic subgroup generated by $M_{f,g}$ in $GL(2n, \F_2)$.

It is a well-known fact (see e.g.~\cite{jacobson85,mullen13}) that the order of the general linear group $GL(2n, \F_2)$, or equivalently its cardinality, is equal to:
\begin{equation}
\label{eq:gl}
|GL(2n,\F_2)| = (2^{2n}-1) (2^{2n} - 2) (2^{2n} - 2^2) \cdots (2^{2n} - 2^{2n-1}) \enspace .
\end{equation}

Let us now recall \emph{Lagrange's theorem}~\cite{gallian12}: \emph{the order of any subgroup $H \le G$ of a finite group $G$ must divide the order of $G$}. Consequently, when determining the order of an invertible Sylvester matrix $M_{f,g} \in GL(2n, \F_2)$, there is no need to consider all powers $t \in \{1,\cdots, |GL(2n, \F_2)|\}$ and check what is the minimum value such that $M_{f,g}^{t} = I_{2n}$. Rather, one has to check this condition only among the \emph{divisors} of $|GL(2n,\F_2)|$. Moreover, it follows that the maximum order attainable by such a Sylvester matrix is $2^{2n}-1$. Indeed, we know that the maximum period reachable by a pair of OCA can be at most $2^{2n}$, due to the fact that the phase space $\F_2^{2n}$ of $\mathcal{S}$ is composed of $2^{2n}$ elements. According to Equation~\eqref{eq:gl}, the maximum order which is both a divisor of $|GL(2n,\F_2)|$ and which is less than or equal to $2^{2n}$ is $2^{2n}-1$. It is not difficult to see that any multiplication of two or more factors occurring in Equation~\eqref{eq:gl} excluding $2^{2n}-1$ is always greater than $2^{2n}$, and thus it cannot represent a valid maximum order for an invertible Sylvester matrix. To summarize, we have proved the following:
\begin{theorem}
\label{thm:max}
Let $d \in \N$, $n=d-1$ and $\mathcal{S} = \langle \F_2^{2n}, H \rangle$ be the dynamical system where $H$ is defined as in Equation~\eqref{eq:sys-upd}, with OCA $F,G: \F_2^{2n} \to \F_2^n$ generated by a pair of linear bipermutive rules $f,g: \F_2^d \to \F_2$. Then, the period $p$ of any state $s \in \F_2^{2n}$ is at most $p \le 2^{2n} - 1$.
\end{theorem}


As an application of the results above, we present a combinatorial algorithm to enumerate all linear OCA pairs of diameter $d$ with maximal period:
\begin{description}
\item[{\sc Enumerate-Maximal-Lin-OCA}$(d)$]
\item[Initialization:] Set $n=d-1$, $t = 2^{2n}-1$, $C = |GL(2n, \F_2)|$
\item[Loop:] For each pair of polynomials $P_f(X), P_g(X) \in \F_2[X]$ with degree $n$ and nonzero constant term do:
  \begin{itemize}
    \item[{\bfseries If}] $\gcd(P_f(X), P_g(X)) = 1$ {\bfseries then}
    \begin{description}
        \item[{\bfseries If}] $M_{f,g}^t = I_{2n}$ AND $M_{f,g}^e \neq I_{2n}$ for all $e|C$, $e < t$ {\bfseries then}
        \begin{itemize}
        \item[--] Print the pair $P_f(X), P_g(X)$
        \end{itemize}
        \item[{\bfseries End If}]
    \end{description}
    \item[{\bfseries End If}]
  \end{itemize}
\item[End Loop]
\end{description}
The algorithm visits all pairs of binary polynomials of degree $n=d-1$ and nonzero constant term, and first checks if the associated polynomials are relatively prime (or equivalently if the two rules form an OCA pair) by computing their greatest common divisor. If this is the case, the algorithm verifies whether the Sylvester matrix has maximal order $t=2^{2n-1}$. By Lagrange's theorem, this operation is accomplished by checking that $M_{f,g}^t = I_{2n}$ and $M_{f,g}^e$ is \emph{not} the identity matrix for any divisor $e$ of the order of  $GL(2n,\F_2)$ less than $t$. If this condition is satisfied, the pair of polynomials $P_f(X), P_g(X)$ is printed.

We applied this algorithm to enumerate all maximal period linear OCA pairs up to diameter $d=11$. For each value of $d$, Table~\ref{tab:enum} reports the number of linear OCA pairs ($\#\mathcal{LOCA}_d$, taken from~\cite{mariot20}), the numbers of pairs with maximal period $2^{2n}-1$ ($\#m\mathcal{LOCA}_d$) and the time required to enumerate them. In particular, we implemented the algorithm in Java and performed the experiments on a 64-bit Linux machine with a 16-core AMD Ryzen processor running at 3.5 GHz and 48 GB of RAM.
\begin{table}[t]
\centering
\caption{Number of maximal period linear OCA pairs of diameter $d\le 11$.}
\begin{tabular}{cccccc}
\hline\noalign{\smallskip}
$d$ & $n$ & $2^{2n}-1$ & $\#\mathcal{LOCA}_d$ & $\#m\mathcal{LOCA}_d$ & Time \\
\noalign{\smallskip}\hline
2 & 1 & 3 & 0 & $-$ & $-$ \\
3 & 2 & 15 & 1 & 1 & $<$ 1s \\
4 & 3 & 63 & 5 & 1 & $<$ 1s \\
5 & 4 & 255 & 21 & 3 & $<$ 1s \\
6 & 5 & 1023 & 85 & 15 & $<$ 1s \\
7 & 6 & 4095 & 341 & 42 & 3.967s \\
8 & 7 & 16383 & 1365 & 181 & 59.162s \\
9 & 8 & 65535 & 5461 & 572 & 18m59.302s \\
10 & 9 & 262143 & 21845 & 1872 & 5h56m10.208s \\
11 & 10 & 1048575 & 87381 & 5899 & 4d16h27m22.126s \\
\hline
\end{tabular}
\label{tab:enum}
\end{table}
As a general remark, it can be noticed that the time required to run {\sc Enumerate-Maximal-Lin-OCA} grows quite rapidly, with more than 4 days required to sift through all pairs of linear bipermutive rules of diameter $11$. Indeed, the most time-consuming step is the computation of the period of the Sylvester matrix. Although we used the {\sc Square-and-Multiply} algorithm~\cite{knuth14} to efficiently exponentiate the matrix, this operation still needs to be performed for all divisors of $2^{2n}-1$ (which is still a significant reduction rather than checking all exponents smaller than $2^{2n}-1$). Also, a second observation is that the number of pairs reaching maximal period $2^{2n}-1$ seems to represent a small subset of linear OCA. In particular, remark that the fourth and fifth columns of Table~\ref{tab:enum} are normalized up to the symmetries considered in~\cite{mariot17}.

\section{Open Problems and Future Directions}
\label{sec:outro}
The theoretical results and the empirical findings of the previous section prompt us with several open problems and directions for further research on maximal period sequences generated with OCA. To begin with, it would be interesting to find a recurrence equation to count all linear OCA pairs with maximal period. Equivalently, this problem amounts to count the number of invertible Sylvester binary matrices of size $2n \times 2n$ with maximum order $2^{2n}-1$. Apparently, this problem has not been studied before in the literature, since the sequence corresponding to the fifth column of Table~\ref{tab:enum} is not reported in the OEIS~\cite{sloane}. A second interesting direction for further research would be to find an efficient characterization of linear OCA with maximal period. Indeed, the main limitation of our approach is that it relies on computing the order of a matrix, which is computationally expensive. Yet, we are interested in Sylvester matrices, which have a very specific structure. It may thus be possible that maximal order can be characterized as a property of the polynomials that define the matrix. Finally, more in general, one could broaden the scope of the investigation to characterize linear OCA pairs with smaller periods, and analyze more closely also the periods of nonlinear OCA pairs. The authors of~\cite{mariot17a} already addressed the construction of nonlinear OCA pairs using \emph{Genetic Algorithms} (GA) and \emph{Genetic Programming} (GP). In this regard, one interesting direction would be to consider the maximization of the largest period as a further optimization goal for GA and GP, either in a single-objective or multi-objective setting.

\section*{Appendix: Source Code and Experimental Data}
The source code of the algorithm and the experimental data discussed in this paper are available at \url{https://github.com/rymoah/hip-to-be-latin-square}.

\section*{Acknowledgments}
The author wishes to thank Luca Manzoni and Antonio E. Porreca for a helpful preliminary discussion on the definition of the dynamical system based on orthogonal CA.

\bibliographystyle{abbrv}
\bibliography{bibliography}

\end{document}